\documentclass[12pt]{article}

\usepackage{geometry}
\usepackage[utf8]{inputenc}

\usepackage{amssymb,amsmath,amsthm}
\usepackage{url}

\usepackage{enumerate}

\theoremstyle{plain}
\newtheorem{theorem}{Theorem}
\newtheorem{lemma}[theorem]{Lemma}

\newtheorem{corollary}[theorem]{Corollary}
\theoremstyle{definition} 

\newtheorem{definition}[theorem]{Definition}

\newcounter{claim}
\renewcommand{\theclaim}{\Alph{claim}}
\newenvironment{claim}{\refstepcounter{claim}%
\par\medskip\par\noindent{\it Claim~\theclaim.~}~\rm}%
{\par\smallskip\par}
\newenvironment{subproof}{\par\noindent{\sl Proof of Claim~\theclaim.~}}%
{$\,\triangleleft$\par\medskip\par}

\makeatletter
\def\@gifnextchar#1#2#3{\let\@tempe#1\def\@tempa{#2}\def\@tempb{#3}%
  \futurelet\@tempc\@gifnch}

\def\@gifnch{\ifx\@tempc\@sptoken\let\@tempd\@tempb%
  \else\ifx\@tempc\@tempe\let\@tempd\@tempa\else\let\@tempd\@tempb\fi\fi\@tempd}

\def\SK@set#1{\left\{#1\right\}}
\def\SK@@set#1#2{\{#1\,:\,
    \begin{array}{@{}l@{}}#2\end{array}
\}}
\def\SK@mset#1{\left\{\!\!\left\{#1\right\}\!\!\right\}}
\def\SK@@mset#1#2{\{\!\!\{#1\,:\,
    \begin{array}{@{}l@{}}#2\end{array}
\}\!\!\}}
\def\BIG@set#1{\Big\{#1\Big\}}
\def\BIG@@set#1#2{\Big\{#1\:\Big|\:
    \begin{array}{@{}l@{}}#2\end{array}
\Big\}}
\newcommand{\Set}[1]{\@gifnextchar\bgroup{\SK@@set{#1}}{\SK@set{#1}}}
\newcommand{\Mset}[1]{\@gifnextchar\bgroup{\SK@@mset{#1}}{\SK@mset{#1}}}
\newcommand{\Bigset}[1]{\@gifnextchar\bgroup{\BIG@@set{#1}}{\BIG@set{#1}}}
\makeatother

\newcommand{\refeq}[1]{(\ref{eq:#1})}

\DeclareMathOperator{\cor}{\mathsf{CR}}
\DeclareMathOperator{\dcor}{\mathsf{D-CR}}
\DeclareMathOperator{\wl}{\mathsf{WL}}

\title{On the Expressibility of\\ the Reconstructional Color Refinement}

\author{V.~Arvind\thanks{The Institute of Mathematical Sciences (HBNI)
and Chennai Mathematical Institute, Chennai, India.}, 
Johannes Köbler\thanks{Institut für Informatik, Humboldt-Universität zu Berlin, Germany.},
Oleg Verbitsky${}^\dagger$\,\thanks{Supported by DFG grant KO 1053/8--2. 
On leave from the IAPMM, Lviv, Ukraine.}}

\date{}

\begin{document}

\maketitle

\begin{abstract}
  One of the most basic facts related to the famous Ulam reconstruction conjecture
  is that the connectedness of a graph can be determined by the deck of its vertex-deleted
  subgraphs, which are considered up to isomorphism. We strengthen this result
  by proving that connectedness can still be determined when the subgraphs in the deck are given up to
  equivalence under the color refinement isomorphism test.
  Consequently, this implies that connectedness is recognizable by Reconstruction Graph Neural Networks,
  a recently introduced GNN architecture inspired by the reconstruction conjecture
  (Cotta, Morris, Ribeiro 2021).
\end{abstract}

\section{Introduction}

For a vertex $v$ of a graph $G$, the \emph{vertex-deleted subgraph} $G\setminus v$ is obtained by
removing $v$ along with all incident edges from $G$. The isomorphism type of $G\setminus v$
is sometimes referred to as a \emph{card}, and then the multiset of all cards is called
the \emph{deck} of $G$. The famous Ulam reconstruction conjecture \cite{Ulam60} says that every graph $G$
with more than 2 vertices is, up to isomorphism, reconstructible from its deck.
Though published in 1960, the problem collection \cite{Ulam60} is a descendant
of the much earlier famous \emph{Scottish Book} \cite{Scottish-Book}. According to Harary \cite{Harary74},
the conjecture has already been discussed since 1929 and, about a hundred years later,
it remains a notoriously hard open problem in graph theory.

The subject has a rich literature, and
one of the most frequently taken approaches is examination of how much information about a graph $G$
can be retrieved from its deck. One of the earliest results in this direction
is referred to in \cite{Harary74} as a ``little theorem'': $G$ is connected if and only if
at least two subgraphs in its deck are connected. This implies the following fact
(see also \cite[Theorem 2.2]{Nash-Williams78})

\begin{theorem}[Harary \cite{Harary64}]\label{thm:harary}
  The connectedness of a graph $G$ is determined by the multiset of all
  vertex-deleted subgraphs $G\setminus v$ of~$G$.
\end{theorem}

Recall that the subgraphs $G\setminus v$ are considered up to isomorphism,
which means that the full structural information about them is available.
It is natural to ask which partial information suffices to decide the connectedness of $G$.
For example, the aforementioned ``little theorem'' of Harary shows that it would be just enough
to know the connectedness of each card in the deck. The question addressed in this note
is whether connectedness can still be decided when the subgraphs $G\setminus v$ are considered up to
an important relaxation of the isomorphism relation, namely up to equivalence
under the color refinement test (also known as one-dimensional Weisfeiler-Leman equivalence).

\emph{Color refinement} iteratively computes an isomorphism-invariant coloring of
the vertex set of an input graph $G$.
Let $\cor(G)$ denote the mutisets of colors assigned by the algorithm to the vertices of $G$
(see Section \ref{s:prel} for definitions).
Color refinement \emph{distinguishes} two graphs $G$ and $H$ if $\cor(G)\ne\cor(H)$.
If $\cor(G)=\cor(H)$, the graphs are called \emph{CR-equivalent}.
Since $\cor(G)$ is a graph invariant, CR-equivalence is a coarser equivalence relation
than isomorphism of graphs.

CR-equivalence has natural characterizations discovered in different research contexts.
Two graphs are CR-equivalent exactly when they are equivalent in the
two-variable first-order logic with counting quantifiers
\cite{ImmermanL90}, when they have equally many vertices and share a
common covering \cite{Angluin80}, when they are fractionally
isomorphic \cite{ScheinermanU97}, and when the number of homomorphisms
from every tree to these graphs is the same~\cite{Dvorak10}.

Color refinement is a versatile tool with applications in diverse fields \cite{Angluin80,GroheKMS21,RiverosSS24}.
Most prominently, it is a practical method of isomorphism testing.
The graph invariant $\cor(G)$ can be computed efficiently \cite{ImmermanL90}
and identifies the isomorphism type of $G$ for almost all graphs \cite{BabaiES80}.
In machine learning, color refinement is used for analysis and comparison of
graph-structured data \cite{ShervashidzeSLMB11} and for understanding the expressiveness of graph
neural networks (GNNs)~\cite{MorrisLMRKGFB23}.

Our main result is the following strengthening of Theorem~\ref{thm:harary}.

\begin{theorem}\label{thm:main}
  The connectedness of a graph $G$ is determined
  by the multiset of the invariants $\cor(G\setminus v)$ for all
  vertex-deleted subgraphs of~$G$.
\end{theorem}

Note that the invariant $\cor(G)$ is not powerful enough to decide whether
or not $G$ is connected. Indeed, $\cor(C_6)=\cor(2C_3)$, where $C_n$ denotes
the cycle graph of length $n$, and $2C_3$ stands for the vertex-disjoint union of
two copies of $C_3$. This cuts off the easiest route to Theorem \ref{thm:main}
through the potential use of the ``little theorem'' of Harary and shows that
running color refinement on the deck of a graph can be more beneficial than
running it on the graph alone. Theorem \ref{thm:main} is interesting in several respects,
which we discuss now.

\paragraph{Consequences for GNNs.}
The expressive power of a representative class of GNNs is characterizable
in terms of color refinement \cite{MorrisRFHLRG19}. In the search for a more expressive GNN architecture,
Cotta, Morris, Ribeiro \cite{CottaMR21} introduce the model of \emph{$k$-Reconstruction GNNs},
whose conceptual novelty consists in enhancing the standard GNNs by running them
on the deck of $k$-vertex subgraphs of a graph $G$ rather than on $G$ itself. The simplest instructive example,
highlighted in \cite{CottaMR21}, is given by considering the 6-cycle graph $C_6$.
As we just mentioned, $C_6$ is not identified by $\cor(C_6)$. However, $C_6$
is reconstructible from its deck, which consists of six copies of the 5-path graph $P_5$.
Since $P_5$ is identified by $\cor(P_5)$, the cycle graph $C_6$ is identified by
color refinement applied to its deck. We state an immediate consequence of Theorem \ref{thm:main},
providing further evidence that $k$-Reconstruction GNNs can often be more powerful.

\begin{corollary}
The connecteness of an $n$-vertex graph is recognizable by $(n-1)$-Reconstruction GNNs.
\end{corollary}

\paragraph{Graph reconstruction.}
Theorem \ref{thm:main} shows that a prototypical reconstructible
graph property can be recognized even from an incomplete information about the deck.
The research on the reconstruction conjecture provides many results of this kind,
and we here mention two of them which concern connectedness.
One can determine the connectedness of an $n$-vertex graph from any $\lfloor n/2\rfloor + 2$
of its cards \cite{BowlerBFM11}.
This can also be done using the $k$-deck of a graph whenever $k\ge 9n/10$ \cite{GroenlandJST23}.
Note in this respect that the $k$-deck of a graph determines its $k'$-deck whenever $k'\le k$
and, therefore, the $k$-deck with a smaller $k$ potentially provides less information.
Turning back to Theorem \ref{thm:main}, it is natural to combine it with the logical
characterization of color refinement \cite{ImmermanL90}. As a consequence, the connectedness
of a graph can be determined based on the information about its deck expressible in
the two-variable logic with counting quantifiers.

\paragraph{Isomorphism invariants of graphs.}
Let $\dcor(G)$ denote the multiset of $\cor(G\setminus v)$ over all vertices $v$ of $G$.
This graph invariant is interesting to compare not only with with $\cor(G)$ but
as well with other related invariants.
Color refinement is also commonly known as the one-dimensional Weisfeiler-Leman algorithm.
Weisfeiler and Leman \cite{WLe68} actually introduced the two-dimensional version, exploiting
the same idea as color refinement but coloring pairs of vertices instead of single vertices.
Let $\wl(G)$ denote the graph invariant computed by this algorithm on input $G$.
A higher dimension results in a more powerful isomorphism test and, in particular,
$\cor(G)$ is determined by $\wl(G)$. It is not hard to show that $\dcor(G)$
is also determined by $\wl(G)$. Whether $\cor(G)$, $\dcor(G)$, and $\wl(G)$
form a hierarchy of graph invariants linearly ordered by their strength is an open question.
More precisely, it is unknown whether or not $\cor(G)$ is determined by $\dcor(G)$.
This question is discussed in \cite[Section 6.8]{ScheinermanU97} and \cite[Section 4.1]{CottaMR21}.
As we mentioned above, $\cor(G)$ does not allow us to see whether $G$ is connected.
On the other hand, it is well known that the connectedness of a
graph $G$ can be determined from $\wl(G)$. Theorem \ref{thm:main} provides
a formal strengthening of the last fact by showing that the connectedness of $G$
is determined even by~$\dcor(G)$.

\section{Color refinenemt}\label{s:prel}

The vertex set of a graph $G$ is denoted by $V(G)$.
The \emph{neighborhood} $N_G(x)$
of a vertex $x\in V(G)$ consists of all vertices adjacent to $x$.
We write $\deg_G x=|N_G(x)|$ to denote the degree of~$x$.

On an input graph $G$, the \emph{color-refinement} algorithm (to be
abbreviated as \emph{CR}) iteratively computes a sequence of colorings
$C^r_G$ of $V(G)$.  The initial coloring $C^0_G$ is monochromatic,
i.e., $C^0_G(x)=C^0_G(x')$ for all $x,x'\in V(G)$, where the same initial color is used for all $G$.
Each subsequent color of a vertex $x$ is obtained from the preceding coloring by the rule
$$
C^{r+1}_G(x)=\Mset{C^{r}_G(y)}_{y\in N_G(x)},  
$$
where $\Mset{}$ denotes a multiset.
Note that $C^{1}_G(x)=C^{1}_G(x')$ exactly when $\deg_G(x)=\deg_G(x')$, that is,
the color $C^{1}_G(x)$ can just be seen as the degree of the vertex $x$.

We define $C_G(x)=(C^0_G(x),C^1_G(x),C^2_G(x),\ldots)$ as the sequence of $C^r_G(x)$ for all $r$
and set $\cor(G)=\Mset{C_G(x)}_{x\in V(G)}$. We say that CR
\emph{distinguishes} two graphs $G$ and $H$ if $\cor(G)\ne\cor(H)$.
If $G$ and $H$ are indistinguishable by CR, i.e.,
$\cor(G)=\cor(H)$, they are referred to as \emph{CR-equivalent}.

Defining $C_G(x)$ as an infinite sequence facilitates proving Theorem \ref{thm:main}
in the subsequent sections. For example, this definition is needed for the equivalences
in Lemma \ref{lem:sim-conn} below. Color refinement is, nevertheless, an absolutely
practical procedure. Indeed, if $n$-vertex graphs $G$ and $H$ are distinguishable by CR,
then they are distinguishable after the $n$-th refinement round at latest in the sense that
$\Mset{C^n_G(x)}_{x\in V(G)}\ne\Mset{C^n_H(x)}_{x\in V(H)}$.
Note in this respect that the color names of $C^{r}_G(x)$ increase
exponentially with $r$ if they are encoded in a straightforward way. This can
be avoided by renaming the colors synchronously on $G$ and $H$ after
each round.

\section{Useful lemmas}

As $C^{1}_G(x)=\deg_G x$, we have $C^{2}_G(x)=\Mset{\deg_G(y)}{y\in N_G(x)}$, which
can be referred to as the \emph{iterated degree} of $x$.
We begin with a particularly useful fact about the reconstructibility of iterated degrees.

\begin{lemma}[{Nash-Williams \cite[Lemma 3.3]{Nash-Williams78}}]\label{lem:GHuv}
  Suppose that graphs $G$ and $H$ have the same deck up to CR-equivalence, that is,
  $$\Mset{\cor(G\setminus x)}_{x\in V(G)}=\Mset{\cor(H\setminus y)}_{y\in V(H)}.$$
  If $\cor(G\setminus u)=\cor(H\setminus v)$
  for a vertex $u\in V(G)$ and a vertex $v\in V(H)$, then $C^2_G(u)=C^2_H(v)$.
\end{lemma}

If two graphs have different number of vertices, then they are
straightforwardly distinguished by CR. Nevertheless, such a pair of graphs
might still be similar in the sense that the CR colors occurring in
them are the same, and the difference is only in the color
multiplicities. In the proof of Theorem \ref{thm:main}, this
similarity concept will play an important role.

\begin{definition}
  We say that graphs $A$ and $B$ are \emph{CR-similar} and write
  $A\equiv B$ if $\Set{C_A(x)}_{x\in V(A)}=\Set{C_B(x)}_{x\in
    V(B)}$. That is, we require that the \emph{sets} of the colors
  produced by CR on $A$ and $B$ be equal (although the
  \emph{multisets} of these colors might be unequal).
\end{definition}

Note that two graphs with different numbers of vertices can be
CR-similar, for example, $C_3\equiv C_4$.

\begin{lemma}\label{lem:subgraphs}
  \hfill

  \begin{enumerate}[\bf 1.]
  \item
    If $A'$ is a proper subgraph of a connected graph $A$, then $A\not\equiv A'$.
  \item
    Let $A'$ be a proper subgraph of $A$ and $B'$ be a proper subgraph of $B$.
    Suppose that both $A$ and $B$ are connected. Then it is impossible that simultaneously
    $A\equiv B'$ and $B\equiv A'$.
  \end{enumerate}
\end{lemma}

The proof is based on a characterization of the CR-color $C_G(x)$ in terms of
the universal cover of the graph $G$ rooted at the vertex $x$. We first give
the relevant definitions.

Let $G$ and $K$ be connected graphs, where $G$ is finite graph, while $K$ can be
either finite or infinite. Let $\alpha$ be a surjective homomorphism from $K$ onto
$G$. We call $\alpha$ a \emph{covering map} if its restriction to the neighborhood of
each vertex in $K$ is bijective. If such an $\alpha$ exists, we say that
$K$ \emph{covers} $G$. A graph $U$ is called a \emph{universal cover} of
$G$ if $U$ covers every graph covering $G$. A universal cover
$U=U_G$ of $G$ is unique up to isomorphism. It can be seen as an unfolding of $G$
into a (possibly infinite) tree starting from an arbitrarily chosen vertex $x$ of $G$.
Denote this unfolding by $U_{G,x}$ and note that $U_{G,x}\cong U_G$ whatever $x$.
If $G$ is a tree, then $U_{G,x}\cong G$. If $G$ contains a cycle, then $U_{G,x}$ is an infinite tree.

The unfolding $U_{G,x}$ can formally be described as follows. A sequence of vertices
$x_0x_1\ldots x_k$ in $G$ is called a \emph{walk starting at $x_0$} if every two successive vertices $x_i$ and $x_{i+1}$
are adjacent. If $x_{i+1}\ne x_{i-1}$ for all $0<i<k$, then the walk is \emph{non-backtracking}.
Now, $U_x$ is the graph whose vertices are all non-backtracking walks in $G$ starting at $x_0=x$,
with any two walks of the form $x_0x_1\ldots x_k$ and $x_0x_1\ldots x_k x_{k+1}$
being adjacent.

A straightforward inductive argument shows that a covering map $\alpha$
preserves the coloring produced by CR, that is, $C_K^r(u)=C_G^r(\alpha(u))$
for all $r$. This has the following consequence. Let $x$ be a vertex of a connected graph $G$
and $y$ be a vertex of a connected graph $H$. If $U_{G,x}$ and $U_{H,y}$ are isomorphic as rooted trees,
then $C_G(x)=C_H(y)$. The converse implication is also true; see, e.g., Lemmas 2.3 and 2.4 in \cite{KrebsV15}.
We state this equivalence in a more precise form. Let $U_{G,x}^r$ denote the rooted tree $U_{G,x}$
truncated at level~$r$.

\begin{lemma}\label{lem:color-cover}
 $C_G^r(x)=C_H^r(y)$ if and only if $U_{G,x}^r\cong U_{H,y}^r$, where $\cong$ denotes the isomorphism of rooted trees. 
\end{lemma}

The characterization of CR-colors in terms of universal covers readily implies the following equivalences.

\begin{lemma}\label{lem:sim-conn}
  The following conditions are equivalent for connected graphs $G$ and $H$:
  \begin{enumerate}
  \item $G\equiv H$;
  \item
    $U_G\cong U_H$;
  \item
    $\Set{C_G(x)}_{x\in V(G)}\cap\Set{C_H(y)}_{y\in V(H)}\ne\emptyset$.
  \end{enumerate}
\end{lemma}

Lemma \ref{lem:sim-conn} has a simple consequence about the CR-similarity of not necessarily
connected graphs. For graphs $G_1,\ldots,G_s$, we write $G_1+\cdots+G_s$ to denote the
vertex-disjoint union of these graphs.

\begin{lemma}\label{lem:sim-disconn}
  \hfill
  
  \begin{enumerate}[\bf 1.]
  \item If $G\equiv H$, then for every connected component $G'$ of $G$ there is
    a connected component $H'$ of $H$ such that $G'\equiv H'$ (and vice versa).
  \item
    If $G\equiv H$ where $G$ is connected and $H=H'+H''$, then $G\equiv H'\equiv H''$.
  \end{enumerate}
\end{lemma}

\begin{proof}[Proof of Lemma~\ref{lem:subgraphs}]
  Let us extend the notation $U_{G,x}^r$ to not necessarily connected graphs
  by setting $U_{G,x}^r=U_{G',x}^r$ where $G'$ is the connected component of $G$ containing
  the vertex $x$.
  
  \textit{1.}  Let the positive integer $D$ be one more than the
  diameter of $A$.  Fix a vertex $u$ of the subgraph $A'$ maximizing
  the number of vertices in the truncated tree $U_{A',u}^D$. The
  description of a universal cover in terms of non-backtracking walks
  readily implies that $U_{A,u}^D$ has strictly more vertices than
  $U_{A',u}^D$ and, hence, than $U_{A',w}^D$ for any vertex $w$ of
  $A'$.  Indeed, our choice of $D$ ensures that there is a walk of
  length at most $D$ starting from $u$ and passing through an edge of
  $A$ absent in $A'$.  Lemma~\ref{lem:color-cover}, therefore, implies
  that $C_A(u)\ne C_{A'}(w)$ for any $w\in V(A')$.

  \textit{2.}  Towards a contradiction, assume that $A\equiv B'$ and
  $B\equiv A'$.  Let $D$ be the positive integer that is one more than
  the larger of the diameters of $A$ and $B$.  Choose a vertex $u$ of
  the subgraph $A'$ maximizing the number of vertices in the truncated
  tree $U_{A',u}^D$. As in the proof of the first part, note that
  $U_{A,u}^D$ has strictly more vertices than $U_{A',u}^D$ and, hence,
  than $U_{A',w}^D$ for any vertex $w$ of $A'$. The assumption
  $B\equiv A'$ implies by Lemma \ref{lem:color-cover} that every tree
  $U_{A',w}^D$ for $w\in V(A')$ has an isomorphic mate $U_{B,w'}^D$
  for some $w'\in V(B)$ and vice versa. Therefore, $U_{A,u}^D$ has
  strictly more vertices than $U_{B,w}^D$ for any vertex $w$ of
  $B$. The same argument with the roles of $A$ and $B$ interchanged
  provides us with a vertex $v$ in $B$ which has strictly more
  vertices than $U_{A,w}^D$ for any vertex $w$ of $A$, in particular,
  than $U_{A,u}^D$. This contradiction completes the proof.
\end{proof}

\section{Proof of Theorem \ref{thm:main}}

Let $G$ be a connected graph and $H$ any disconnected graph.  We have
to prove that the multisets $\Mset{\cor(G\setminus x)}_{x\in V(G)}$ and
$\Mset{\cor(H\setminus y)}_{y\in V(H)}$ are unequal.  We will assume that these
multisets are equal and derive a contradiction.

Let $H_1,\ldots,H_s$, where $s\ge2$, be the connected components of $H$.
Let $V=V(H)$ and $V_i=V(H_i)$ for $i\le s$.
Identify $V(G)$ with $V$ so that
\begin{equation}
  \label{eq:CRGH}
\cor(G\setminus v)=\cor(H\setminus v)
\end{equation}
for all $v\in V$.
By Lemma \ref{lem:GHuv}, this yields
$$
C^2_G(v)=C^2_H(v)
$$
for all $v\in V$. Taking into account the decomposition of $H$ into connected components, we have
\begin{equation}
  \label{eq:CGHvVi}
C^2_G(v)=C^2_{H_i}(v) \text{ if } v\in V_i.
\end{equation}

Denote the set of articulation (or cut) vertices of $G$ by $A$, and set $B=V\setminus A$.
Let $B_i=B\cap V_i$. If $v\in B_i$, then $G\setminus v$ is connected and
$H\setminus v=H_1+\ldots+H_i\setminus v+\ldots+H_s$. By Part 2 of Lemma \ref{lem:sim-disconn},
Equality \refeq{CRGH} implies that
\begin{equation}
  \label{eq:GvHiHj}
  G\setminus v\equiv H_i\setminus v\equiv H_j\text{ if }v\in B_i\text{ and }j\ne i.
\end{equation}
Applying the similarity $G\setminus v\equiv H_j$ to any two vertices in $B_i$,
we derive the similarity $G\setminus v\equiv G\setminus v'$ whenever $v,v'\in B_i$.
Applying Lemma \ref{lem:GHuv} in the particular case of $G=H$, we conclude that
\begin{equation}
  \label{eq:CGvvBi}
C^2_G(v)=C^2_G(v')\text{ for all }v,v'\in B_i.
\end{equation}

Let $i\ne j$ and assume that $B_i\ne\emptyset$ and $B_j\ne\emptyset$.
Let $x\in B_i$ and $y\in B_j$. By \refeq{GvHiHj}, we obtain
$H_i\setminus x\equiv H_j$ and $H_j\setminus y\equiv H_i$, which contradicts Part 2 of Lemma \ref{lem:subgraphs}.
Therefore, only one of the sets $B_i$ is not empty, that is, $B=B_i$ for some $i$.
Without loss of generality, we suppose that $B=B_1$.

\begin{claim}\label{cl:}
  $B_1=V_1$.
\end{claim}

\begin{subproof}
  Towards a contradiction, assume that $V_1$ contains a cut vertex $v$ of $G$.
  By Part 1 of Lemma \ref{lem:sim-disconn}, Equality \refeq{CRGH} implies that
  the disconnected graph $G\setminus v$ has a connected component $G'$
  such that $G'\equiv H_2$. Choose a vertex $x\in B$ in a connected component of $G\setminus v$
  different from $G'$. Since $x\in B_1$, the similarity \refeq{GvHiHj} implies that
  $G\setminus x\equiv H_2$ and, therefore, $G\setminus x\equiv G'$. However, $G'$ is a subgraph of the connected graph $G\setminus x$,
  contrary to Part 1 of Lemma~\ref{lem:subgraphs}.
\end{subproof}

By Claim \ref{cl:} and Equality \refeq{CGvvBi}, we have
$$
C^2_G(x)=C^2_G(x')\text{ for all }x,x'\in V_1.
$$
Taking into account Equality \refeq{CGHvVi}, this implies that
\begin{equation}
  \label{eq:CH1vv}
C^2_{H_1}(x)=C^2_{H_1}(x')\text{ for all }x,x'\in V_i.
\end{equation}
In particular, $H_1$ is a regular graph, say, of degree~$d$.

Consider the block-cut tree of $G$. Recall that it describes the decomposition
of $G$ into biconnected components, that is, maximal biconnected subgraphs, which are also
called blocks. The vertex set of the tree consists of the blocks and the cut vertices of $G$.
A block $M$ is adjacent to a cut vertex $c$ if $c\in M$.

Let $L$ be a block that is a leaf in the block-cut tree of $G$.
Let $a$ be the single cut vertex of $G$ belonging to $L$.
By Claim \ref{cl:}, $a$ belongs to $A=V\setminus B=V_2\cup\ldots\cup V_s$,
and the other vertices of $L$ belong to $B=B_1=V_1$.
Without loss of generality, suppose that $a\in V_2$.
Let $x$ be a vertex in $L$ adjacent to $a$. Since $x\in B_1$, the similarity \refeq{GvHiHj} implies that
\begin{equation}
  \label{eq:GxH1H2}
  G\setminus x\equiv H_1\setminus x\equiv H_2.
\end{equation}

Since $H_1$ is a $d$-regular graph, every vertex of $H_1\setminus x$ has degree $d-1$ or $d$.
By the similarity \refeq{GxH1H2}, the neighbors of $a$ in $H_2$ are of degree $d-1$ or $d$.
By Equality \refeq{CGHvVi}, the same holds true for $a$ in $G$.
Therefore, every connected component $G'$ of $G\setminus a$ has a vertex of degree $d-2$ or $d-1$
(namely a neighbor of $a$ in $G$). As a consequence, $G'\not\equiv H_1$ for every $G'$. On the other hand,
Equality \refeq{CRGH} implies that
$$
G\setminus a\equiv H_1+H_2\setminus a+\cdots+H_s.
$$
By Part 1 of Lemma \ref{lem:sim-disconn}, we conclude that $H_1\equiv G'$ for some component $G'$ of $G\setminus a$.
This contradiction completes the proof of Theorem~\ref{thm:main}.


\end{document}